\documentclass[a4paper,12pt]{article}

\usepackage[english]{babel}
\usepackage[utf8x]{inputenc}
\usepackage[T1]{fontenc}

\usepackage[letterpaper,margin=0.85in]{geometry}

\usepackage{amsmath,color}
\usepackage{amsthm}
\usepackage{amsfonts}
\usepackage{amssymb}
\usepackage{graphicx}
\usepackage{overpic}
\def\R{\mathbb{R}}

\def\N{\mathbb{N}}

\theoremstyle{plain}
\newtheorem{theorem}{Theorem}[section] 

\theoremstyle{definition}

\theoremstyle{remark}

\def\cC{\mathcal{C}}

\def\cO{\mathcal{O}}

\def\txta{{\textnormal{a}}}

\def\txtd{{\textnormal{d}}}

\def\txtD{{\textnormal{D}}}

\def\txtst{{\textnormal{st}}}

\def\ra{\rightarrow}
\def\I{\infty}

\newcommand{\be}{\begin{equation}}
\newcommand{\ee}{\end{equation}}
\newcommand{\benn}{\begin{equation*}}
\newcommand{\eenn}{\end{equation*}}
\newcommand{\bea}{\begin{eqnarray}}
\newcommand{\eea}{\end{eqnarray}}
\newcommand{\beann}{\begin{eqnarray*}}
\newcommand{\eeann}{\end{eqnarray*}}

\title{Analysis and Predictability for Tipping Points with 
Leading-Order Nonlinear Terms}
\author{Francesco 
Romano\footnote{Faculty of Mathematics, Technical University of 
Munich, Boltzmannstr.~3, 85747 Garching b.~Munich, Germany \& 
Ludwig-Maximilians-Universitaet, Elite Graduate Course Theoretical 
and Mathematical Physics, Theresienstr. 37, 80333, Munich, Germany.}~~and~Christian 
Kuehn\footnote{Faculty of Mathematics, Technical University of 
Munich, Boltzmannstr.~3, 85747 Garching b.~Munich, Germany}}

\begin{document}
\maketitle

\begin{abstract}
Tipping points have been actively studied in various applications as well as from a mathematical viewpoint. A main technique to theoretically understand early-warning signs for tipping points is to use the framework of fast-slow stochastic differential equations. A key assumption in many arguments for the existence of variance and auto-correlation growth before a tipping point is to use a linearization argument, i.e., the leading-order term governing the deterministic (or drift) part of stochastic differential equation is linear. This assumption guarantees a local approximation via an Ornstein-Uhlenbeck process in the normally hyperbolic regime before, but sufficiently bounded away from, a bifurcation. In this paper, we generalize the situation to leading-order nonlinear terms for the setting of one fast variable. We work in the quasi-steady regime and prove that the fast variable has a well-defined stationary distribution and we calculate the scaling law for the variance as a bifurcation-induced tipping point is approached. We cross-validate the scaling law numerically. Furthermore, we provide a computational study for the predictability using early-warning signs for leading-order nonlinear terms based upon receiver-operator characteristic curves.  
\end{abstract}

\textbf{Keywords:} critical transition, tipping point, warning sign, scaling 
law, bifurcation, fast-slow system, stochastic differential equation, ROC curve,
predictability.

\section{Introduction}
\label{sec:intro}

Tipping points (or critical transitions) have been studied intensively in recent
years with a focus on finding early-warning 
signs~\cite{AshwinWieczorekVitoloCox,KuehnCT2,Schefferetal}. One key idea to predict 
a transition is to exploit the effect of critical slowing down indirectly via 
observing a noisy time series of a dynamical system. The idea goes back (at least) 
to the work of Wiesenfeld~\cite{Wiesenfeld1} but has gained recent popularity in 
many contexts, particularly in ecology~\cite{CarpenterBrock} and climate 
science~\cite{DitlevsenJohnsen}. In terms of a fast-slow stochastic differential
equation, the simplest class of examples are systems of the form
\be
\label{eq:fs1}
\begin{array}{lcl}
\txtd u &=& f(u,v)~\txtd t+\sigma~\txtd W,\\
\txtd v &=& \varepsilon ~\txtd t,
\end{array}
\ee
where $u=u(t),v=v(t)\in\R$, $W$ is a one-dimensional Brownian motion, $\sigma>0$ controls
the noise level, and $\varepsilon>0$ is a small parameter. Note that $u$ is a 
fast variable in comparison to the slow variable $v$ as $\varepsilon$ is small.
If one wants to model the simplest situations, when a bifurcation-induced tipping 
occurs, one usually selects for the drift term $f(u,v)$ a normal 
form~\cite{GH,Kuznetsov} for a bifurcation such as $f(u,v)=v-u^2$ for the 
fold or $f(u,v)=uv-u^3$ for the (sub-critical) 
pitchfork~\cite{BerglundGentz}. Even many higher-dimensional cases have
been analyzed by now~\cite{KuehnCT2} for fast-slow SODEs. Let us suppose that
the drift term has a non-hyperbolic steady state at $(u,v)=(0,0)$, or alternatively
formulated the normal hyperbolicity of the critical manifold 
\benn
\cC_0=\{(u,v)\in\R^2:f(u,v)=0\}
\eenn 
breaks down at the origin. Furthermore, assume that the critical manifold has one component, 
which is attracting for the fast dynamics and locally parametrized by 
$\cC_0^\txta=\{u=h(v)\}\subset \cC_0$ for some smooth function $h$ and $(0,0)$
lies on the boundary of $\cC_0^\txta$. The standard tool to understand the local 
fluctuations of the stochastic process $u$ near the origin is now to consider 
the linearized non-autonomous system along $\cC_0^\txta$
\be
\label{eq:linfs}
\txtd U = \txtD_uf(h(\varepsilon t),\varepsilon t)U~\txtd t + \sigma~\txtd W=:
A(\varepsilon t)U~\txtd t + \sigma ~\txtd W.
\ee
Of course, \eqref{eq:linfs} is just a standard one-dimensional non-autonomous
Ornstein-Uhlenbeck (OU) process. In the quasi-steady (or adiabatic) limit 
$\varepsilon~\ra 0$, the process becomes autonomous and can be viewed as a 
parametrized family since the variable $v$ is fixed and can then be 
viewed as a parameter $v$; to emphasize when this viewpoint is taken we
shall write $p=v$. The solution of the resulting OU process is easy
to calculate~\cite{Gardiner}. If we let $V_\I=\lim_{t\ra \I}\textnormal{Var}(U(t))$ 
be the time-asymptotic variance then one finds for the fold and pitchfork 
examples above
\be
V_{\I,\textnormal{fold}}=\cO(p^{-1/2}),\qquad V_{\I,\textnormal{pitchfork}}
=\cO(p^{-1})\qquad \text{as $p\nearrow 0$,}
\ee
i.e., the linearized leading-order approximation of the variance of the 
process $u$ diverges with certain universal exponents as $p$ tends to the 
bifurcation point. Note that the linear approximation only holds for 
sufficiently small noise and breaks down for the
system with $0<\varepsilon\ll 1$ in a very small $\varepsilon$-dependent 
neighbourhood of the origin~\cite{KuehnCT2} but it provides a very good 
approximation otherwise. Hence, variance growth can often be used as an
early-warning sign for bifurcation-induced tipping. However, we did 
make the key assumption that linear terms are of leading-order.
In this paper, we study leading-order nonlinear terms, which preclude
the use of results from linear stochastic processes.\medskip 

In Section~\ref{sec:back} we provide the mathematical background 
and framework for our setting. In Section~\ref{sec:variance}, we prove
a variance scaling law for polynomial nonlinearities $pU^k$ ($k$ odd) 
and cross-validate it numerically. The universal scaling exponent can 
be computed explicitly and divergence of the variance is given by 
$$V_{\I,\textnormal{nonlin}}=\cO(p^{-2/(k+1)})\qquad \text{as }p\nearrow 0.$$
In Section~\ref{sec:numerics}, we provide a computational 
study to better understand practical predictability for leading-order nonlinear 
terms using receiver-operator characteristic (ROC) 
curves~\cite{KuehnZschalerGross,BoettingerHastings,ZhangHallerbergKuehn} in 
comparison to the linear case and also depending upon sliding window length, lead 
time, and alarm volume size. 

\section{Background and Framework}
\label{sec:back}

Consider the following ordinary differential equation (ODE)
depending on the parameter $p\in\R$
\begin{equation}
\label{eqn:nonlinode}
\frac{\txtd U}{\txtd t}=p \, U^k, \; \; \; \; U=U(t)\in\R,~ U_0:=U(0),
\end{equation}
and assume $k\in \N$ to be odd. The point $U_*=0$ is a steady state 
for~\eqref{eqn:nonlinode}. One easily checks using the gradient structure
of one-dimensional ODEs that $U_*$ is (even globally) stable for $p < 0$ and 
unstable for $p > 0$. In particular, \eqref{eqn:nonlinode} has a bifurcation,
respectively a bifurcation-induced tipping, when $p=0$. Since we are interested
in early-warning signs in the SODE case, we now study
\begin{equation}
\label{eqn:nonlinsde}
\txtd U =p \, U^k ~\txtd t + \sigma ~\txtd W, \; \; \; \; U(0)=U_0,
\end{equation}
where $\sigma>0$, $W$ a one-dimensional Brownian motion on a filtered 
probability space $(\mathbb{R}, \mathcal{F}, \mathcal{F}_t, \mathbb{P})$ 
and $U_0$ is a $\mathcal{F}_0$-measurable random variable. In the following, 
we are going to show that the variance of the (unique global) solution $U(t)$ 
to~\eqref{eqn:nonlinsde} has a divergent behavior as $p\nearrow 0$. We are going
to exploit the Fokker-Planck equation to find an explicit expression for the 
asymptotic variance 
\benn
V_\infty:=\lim_{t \rightarrow \infty} \textrm{Var}(U(t)) 
\eenn
in Theorem~\ref{thm:ratedivergence}. First, we provide some background. 
The SODE~\eqref{eqn:nonlinsde} has a unique global-in-time solution up to equivalence 
for any odd $k$. 

\begin{theorem}
\label{thm:exuni}
For $p<0$ and any $t>0$, the stochastic process 
$$U(t)=U_0 + p \int_0^t U^k(s)~\txtd s + \sigma~W(t)$$ 
is the unique solution (up to equivalence) to~\eqref{eqn:nonlinsde}. 
\end{theorem}

\begin{proof}
According to \cite[Thm.~3.5]{Khasminskii1}, it is enough to prove that 
there exists a non-negative $C^{1,2}$ function $\psi $ on $[0, \infty) \times 
\mathbb{R}^m$ such that for some constant $c > 0$
\benn
L\psi  \leq c\psi  \quad \text{and}\quad 
\psi_R = \inf_{|x|>R} \psi(t,x) \rightarrow \infty \textrm{ as } R \rightarrow \infty,
\eenn
where $$L\psi(s,x)=\partial_s \psi(s,x) + p X_s^k \partial_x \psi(s,x) 
+ \frac{\sigma^2}{2} \partial_{xx}\psi(s,x).$$
We set $\psi(s,x)=(x^2+1)^a$ for $a>1$. $\psi$ is obviously $C^{1,2}$ and it 
satisfies $$ \psi_R = \inf_{|x|>R} \psi(t,x) = \inf_{|x|>R} (x^2+1)^a = 
(R^2 + 1)^a \rightarrow \infty \textrm{ as } R \rightarrow \infty.$$ 
It is only left to prove that for some $c > 0$ it holds $L\psi \leq c \psi$. 
We compute $L\psi$ to obtain:
\begin{align*}
L\psi(s,x)&=\partial_s (x^2+1)^a + p x^k \partial_x (x^2+1)^a 
+ \frac{\sigma^2}{2} \partial_{xx} (x^2+1)^a \\
&= 2pa x^{k+1}(x^2+1)^{a-1} + a \sigma^2 (x^2+1)^{a-1}+ 2a(a-1) 
\sigma^2 x^2(x^2+1)^{a-2}  \\
& \leq a \sigma^2 (x^2+1)^{a-1}+ 2a(a-1) \sigma^2 
x^2(x^2+1)^{a-2}  \\
& \leq a \sigma^2 (x^2+1)^{a-1}+ 2a(a-1) \sigma^2 (x^2+1)^{a-1} \\
&= [a \sigma^2 + 2a(a-1) \sigma^2] (x^2+1)^{a-1}  \\
& \leq [a \sigma^2 + 2a(a-1) \sigma^2] (x^2+1)^a = [a \sigma^2 + 
2a(a-1) \sigma^2] \psi(s,x),
\end{align*}
where we used $p < 0, a >0$ and the fact that $k+1$ is even. Hence, the 
claim follows. 
\end{proof}

We recall that, under certain conditions, solutions to SODEs are Markov 
processes and under stronger assumptions their distribution converge in 
time to a stationary distribution, which can be identified with the solution 
to Fokker-Planck equation. Specifically, the following holds 
(see \cite[Sec.~4.4-4.7, Lem.~4.16]{Khasminskii1}):

\begin{theorem}
\label{thm:convergencemarkov}
Consider a stochastic differential equation of the form 
\begin{equation}
\txtd U=g(U)~\txtd t+\sigma~\txtd W,\qquad U=U(t)\in\R.
\end{equation}
Suppose there exists a bounded open domain $\Omega \subset 
\mathbb{R}$ with regular boundary $\Gamma$ such that
\begin{enumerate}
\item If $x \in \R\setminus\Omega$, the mean time $\tau$ at which 
a path starting from $x$ reaches the set $\Omega$ is finite,
\item $\sup_{x \in K}\mathbb{E}^x [\tau] < \infty$ for every compact 
set $K \subset \mathbb{R}$.
\end{enumerate}
Then, the Markov process $U=U(t)$ has a unique stationary distribution 
$\mu$ and, independently of the distribution of $U_0$, the distribution 
of $U$ converges to $\mu$ as $t \rightarrow +\infty$. Moreover, 
$\mu(A)$ has stationary density $\rho^{\txtst}(x)$ with respect to Lebesgue 
measure, given by the unique (normalized) bounded solution of the 
stationary Fokker-Planck equation 
\begin{equation}
\label{eqn:fokpla}
L^*\rho^{\txtst}:=\frac{\sigma^2}{2} \partial_{xx} \rho^{\txtst}(x)
- \partial_{x}(f(x)\rho^{\txtst}(x))=0.
\end{equation}
\end{theorem}

If we can apply Theorem~\ref{thm:convergencemarkov}, and if we can
compute the stationary solution $\rho^{\txtst}$ and from it the variance,
then we can circumvent any OU processes used for the linear case.

\section{Asymptotic Result for the Variance}
\label{sec:variance}

We now show that Theorem~\ref{thm:convergencemarkov} can be used to derive an asymptotic 
result for the variance of~\eqref{eqn:nonlinsde}:

\begin{theorem}[\textbf{variance scaling law}]
\label{thm:ratedivergence}
Suppose $p<0$ and consider the one-dimensional nonlinear SDE 
\begin{equation}
\label{eqn:soldiv}
\txtd U =p \, U^k~ \txtd t + \sigma~\txtd W, \; \; \; \; U(0)=U_0
\end{equation}
where $U_0$ is an $\mathcal{F}_0$-measurable random variable. For each odd $k\in \N$, 
the associated deterministic ODE has a bifurcation in $p=0$. Consider the stationary 
distribution $\rho^{\txtst}$ of the solution $U_t$ to \eqref{eqn:soldiv} and denote 
its variance by $V_\infty$. The following holds for all odd $k\in\N$: 
\begin{equation}
V_\infty = \Big (- \frac{k+1}{2p} \Big )^{2/(k+1)} \frac{\Gamma(
1+3/(k+1))}{\Gamma(1+1/(k+1))},
\end{equation}
where $\Gamma$ is the usual Gamma function.
In particular, for all odd $k$, the asymptotic behavior as $p\nearrow 0$ is given by
\begin{equation}
V_\infty = \cO \Big (\frac{1}{p^{2/(k+1)}}  \Big ).
\end{equation}
\end{theorem}

\begin{proof}
The proof proceed as follows: first, we show that our system satisfies the conditions 
in Theorem~\ref{thm:convergencemarkov}, so that we can use Fokker-Planck equation to 
compute the stationary distribution; then, we compute explicitly the solution to the 
Fokker-Planck equation and its variance to conclude the proof.\medskip

\textit{Step 1: Convergence to the asymptotic distribution.} Fix $\Omega=(-R,R)$, which 
is open and bounded. To check the first condition in Theorem~\ref{thm:convergencemarkov}, 
it is enough to prove by Theorem~\ref{thm:exuni} and \cite[Thm.~3.9]{Khasminskii1} that 
there exists in $[0, + \infty) \times (\R\setminus\Omega)$ a nonnegative function 
$\Psi(s,x) \in C^{1,2}$ such that 
\begin{equation*}
L\Psi(s,x) \leq -\alpha (s),
\end{equation*}
where $\alpha(s) \geq 0$ is a function such that
\begin{equation*}
\beta(t)= \int_0^t \alpha(s)~\txtd s \rightarrow \infty \; \; as \; t \rightarrow \infty.
\end{equation*} 
We choose $\Psi(s,x)=(1+x^2)^a$, $a > 1$, which satisfies the regularity 
hypothesis. Moreover, 
\begin{align*}
L\Psi(s,x)&= 2pa x^{k+1}(x^2+1)^{a-1} + a \sigma^2 (x^2+1)^{a-1}+ 2a(a-1) 
\sigma^2 x^2(x^2+1)^{a-2} \\
&\leq [ pa x^{k+1} + a \sigma^2 + 2a(a-1) \sigma^2] (x^2+1)^{a-1}.
\end{align*}
Choosing $R$ big enough we can guarantee $$pa x^{k+1} + a \sigma^2 + 2a(a-1) 
\sigma^2 < -\upsilon$$ 
for all $x \in (\R\setminus\Omega)=:\Omega^c$ and some constant $\upsilon>0$. 
This implies 
\begin{equation*}
L\Psi(s,x) \leq \upsilon(R^2+1)^{a-1}=:\alpha.
\end{equation*}
Furthermore, we have
\begin{equation*}
\beta(t)=\int_0^t \alpha ~\txtd s = \alpha t \rightarrow \infty \; \; 
\text{as} \; t \rightarrow \infty
\end{equation*}
as required. This proves the first condition in Theorem~\ref{thm:convergencemarkov}. 
For the second condition, by \cite[Thm.~3.9]{Khasminskii1}, the expectation of the random 
variable $\beta(\tau_{\Omega^c})$ exists and satisfies the inequality 
\begin{equation*}
\mathbb{E}^{s,x} [\beta (\tau_{\Omega^c}) ]\leq \beta (s) + \Psi(s,x),
\end{equation*}
which implies
$$\mathbb{E}^{s,x} [\tau_{\Omega^c} ]\leq s + \frac{(1+x^2)^a}{\alpha}.$$ 
Now, setting $s=0$ we have 
\begin{equation*}
\mathbb{E}^{x} [\tau_{\Omega^c}] \leq \frac{(1+x^2)^a}{\alpha} < \infty 
\textrm{ for all compact sets } K.
\end{equation*}
Therefore, Theorem~\ref{thm:convergencemarkov} implies that the density of~$U(t)$ 
converges to $\rho^{\txtst}$ as $t \rightarrow +\infty$, independently of the initial 
condition $U_0$. \medskip

\textit{Step 2: Density computation.} $\rho^{\txtst}$ is the unique bounded (normalized) 
solution to the stationary Fokker-Planck equation
\begin{equation}
\label{eqn:fokkplancknonlinearcase}
0 = - \partial_x [p x^k \rho^{\txtst}(x)] + \frac{1}{2} \partial_{xx}\rho^{\txtst}(x).
\end{equation} 
In our case, one can simply compute by direct integration that
\begin{align*}
2p y^k \rho^{\txtst}(y) &= \partial_y \rho^{\txtst}(y) - \partial_y \rho^{\txtst}(0)
\end{align*}
Assume $\partial_x \rho^{\txtst}(0)=0$. Since we are looking for the unique bounded 
solution of~\eqref{eqn:fokkplancknonlinearcase}, we can justify our assumption a posteriori 
by showing that the solution we obtain is bounded. We solve the last equation and obtain
\begin{equation*}
\rho^{\txtst}(x)=\rho^{\txtst}(0) \exp \Big (\frac{2p}{k+1} x^{k+1} \Big ).
\end{equation*}
Since $p <0$ and $k+1$ is even, the exponential in the above expression can be integrated 
over $\R$. The constant $\rho^{\txtst}(0)$ is the normalizing constant so we get 
\begin{equation*}
\rho^{\txtst}(x)=\frac{\exp \Big (\frac{p}{m} x^{2m} \Big )}{\int_\mathbb{R} 
\exp \Big (\frac{p}{m} x^{2m} \Big )~\txtd x},
\end{equation*}
where $k+1=2m$ and $m \in \mathbb{N}$. This shows in particular that $\rho^{st}$ 
is bounded, as required.\medskip

\textit{Step 3: Asymptotic variance.} Since $\rho^{\txtst}$ is symmetric, its expected 
value is $0$. Its variance is then given by
\begin{align*}
V_\infty &= \frac{1}{\int_\mathbb{R} \exp \Big (\frac{p}{m} x^{m} \Big )~\txtd x} 
\int_\mathbb{R} x^2 \exp \Big (\frac{p}{m} x^{2m} \Big )~\txtd x\\
&= \frac{2/3 (-p/m)^{-3/2m} \Gamma(1+3/2m)}{2/3 (-p/m)^{-1/2m} \Gamma(1+1/2m)} \\
&= \Big (- \frac{1}{p} \Big )^{1/m} \frac{\Gamma(1+3/2m) \;  m^{1/m}}{\Gamma(1+1/2m)}. 
\end{align*}
This concludes the proof.
\end{proof}

We remark that the approach we followed is quite general and it has not much to do with 
specific properties of~\eqref{eqn:soldiv} except it being a scalar equation. For 
higher-dimensional cases, we would have to use approximation and/or reduction 
methods to understand stationary solutions of the Fokker-Planck equation~\cite{Risken}.\medskip

\begin{figure}
\centering
\begin{overpic}[width=0.6\linewidth]{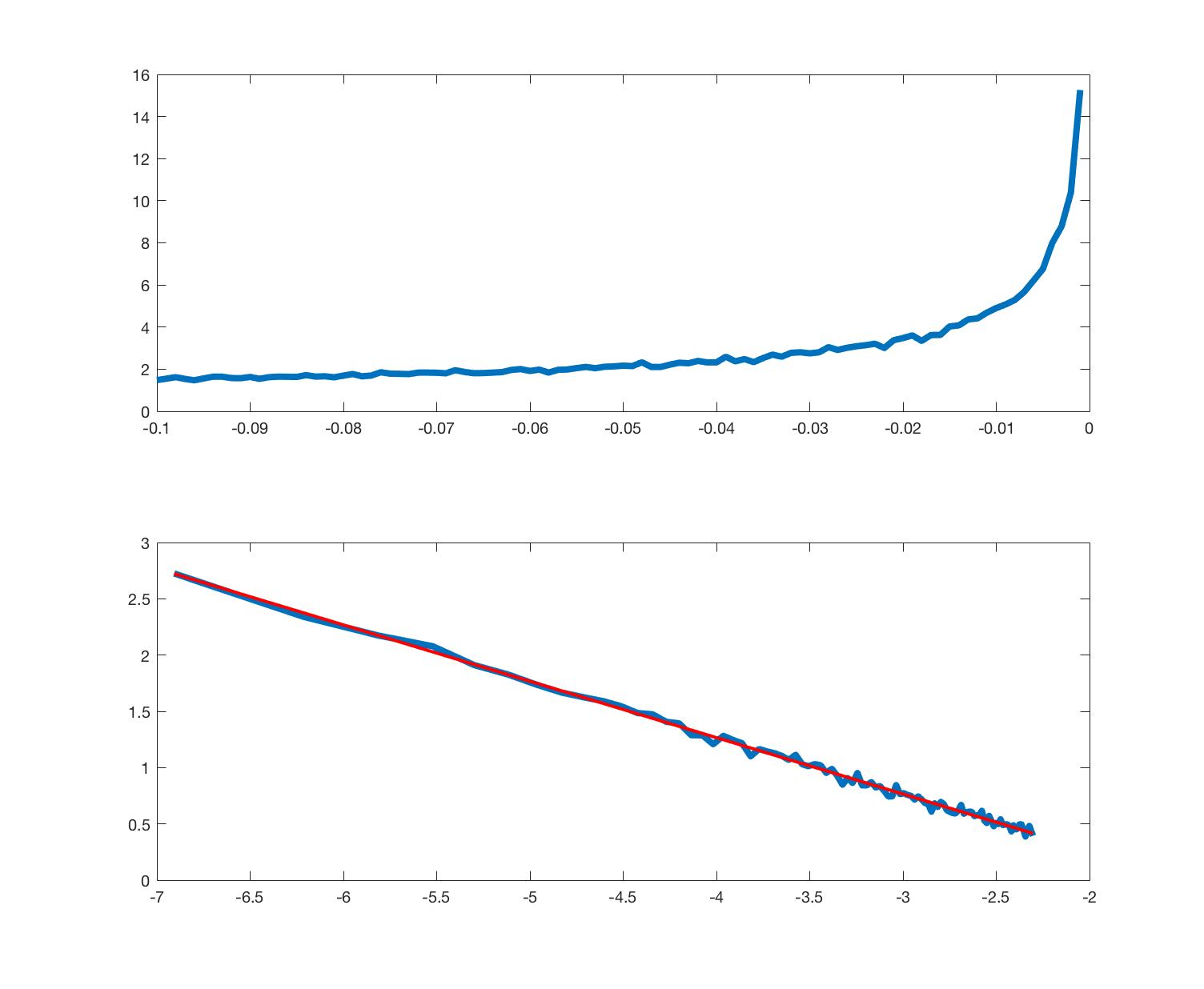}
\put(48,77){k=3}
\put(43,38){m=-0.50085}
\put(8,60){\makebox(0,0){\rotatebox{90}{$V_\infty$}}}
\put(49,44){$p$}
\put(8,23){\makebox(0,0){\rotatebox{90}{$\log(V_\infty)$}}}
\put(47,4){$\log(-p)$}
\end{overpic}
\caption{\label{fig:cov3} For $p \in [-0.1,0]$ we solve the equation $\txtd u = p u^3~\txtd t 
+ \txtd W_t$ on the interval $[0, 100]$ with initial condition $u_0=1$ using the 
Euler-Maruyama approximation method with $N=1000$ time steps. The plot above shows 
the numerically approximated variance $V_\infty$ from $1000$ sample paths and $t^*=100$. 
The loglog plot below (blue) shows clearly the (inverse) polynomial dependence. In red 
we plot the linear interpolation.}
\end{figure}
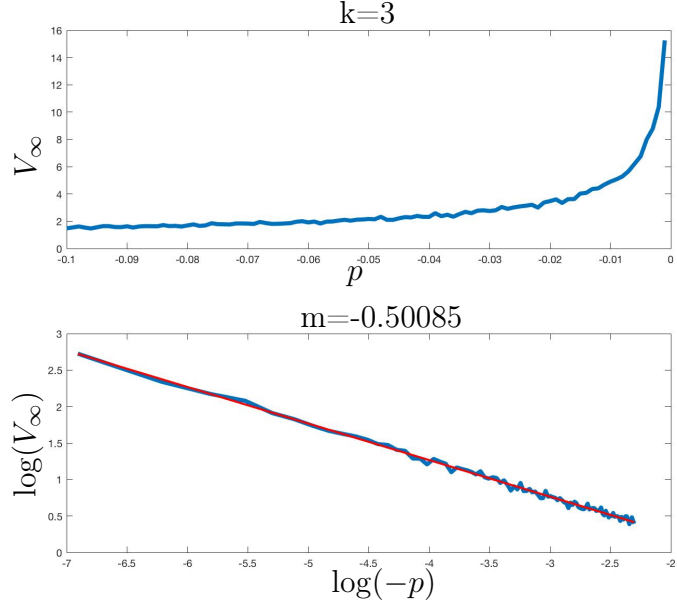

To cross-validate the theoretical result, we plot in Figure~\ref{fig:cov3} a numerical 
approximation of the asymptotic variance for $k=3$ obtained in the following way: 
\begin{enumerate}
\item[(N1)] first, we consider a sequence $p_i$ converging to  the bifurcation 
in $p=0$ as $i \rightarrow \infty$;
\item[(N2)] then, we choose (for each $p_i$) a large enough value $t=t^*$ so 
that the variance can be assumed to be close to the asymptotic limit;
\item[(N3)] finally, using Euler-Maruyama method~\cite{Higham} we simulate a 
large enough number of sample paths to the 
SDE~\eqref{eqn:nonlinsde} so that the empirical variance can be accurately 
computed.
\end{enumerate}

We remark here that the values of $t^*$ and the number of sample paths have been 
chosen empirically via 
numerical simulations to ensure the required conditions to be satisfied. In 
Figure~\ref{fig:cov3} we also show a loglog plot to highlight more clearly the relation 
of the form $$V_\infty = \cO \Big (\frac{1}{p^m} \Big ).$$ The following table shows the 
results for odd values of $k$ between 3 and 11: 

\begin{center}
\begin{tabular}{ |c| c| c| }
\hline
$k$ & $m$ & $2/(k+1)$ \\
\hline
3 & -0.50085 & -0.5 \\
5 & -0.33714 & -0.3333  \\
7 & -0.24961 & -0.25 \\
9 & -0.20656 & -0.2 \\
11 & -0.17154 & -0.1667 \\
\hline
\end{tabular}
\end{center}

As one can see by comparing the second and third column, the numerical results are really 
close to the analytical analysis, so it is also possible to observe the scaling in direct
practical simulations and/or data.

\section{Statistics for early-warning signs}
\label{sec:numerics}

We continue to study~\eqref{eqn:nonlinsde} and want to determine, how well 
statistical classifiers based on our previous findings can be used to predict 
tipping points~\cite{BoettingerHastings,HallerbergKantz,KuehnZschalerGross,
ZhangHallerbergKuehn}. Returning to our model class~\eqref{eq:fs1} we include 
in~\eqref{eqn:nonlinsde} the evolution of the parameter and study
\begin{equation}
\begin{cases}
\txtd U= \frac{p}{\varepsilon}U^k~\txtd s 
+ \frac{\sigma}{\sqrt{\varepsilon}}~\txtd W, \\
\txtd p=\txtd s,
\end{cases}
\end{equation}
where $s=\varepsilon t$, and we can also view the system as a single 
non-autonomous SODE
\begin{equation}
\label{eqn:fastslowunique}
\txtd U(p)= \frac{p}{\varepsilon}U^k(p)~\txtd p + 
\frac{\sigma}{\sqrt{\varepsilon}}~\txtd W(p).
\end{equation}
$\cC_0^{\txta}:=\{U=0, \; p <0 \}$ contains attracting steady states: if the parameter is 
initially negative, the evolution converges to $\cC_0^{\txta}$ (fast dynamics) and then 
remains close to it (slow dynamics) until tipping happens. Simulations 
for different choices of the initial conditions $(p_0, U_0)$, the 
nonlinearity parameters $k$, and the parameters $\sigma$ and $\varepsilon$ are
shown in Figure~\ref{fig:ff3}.  

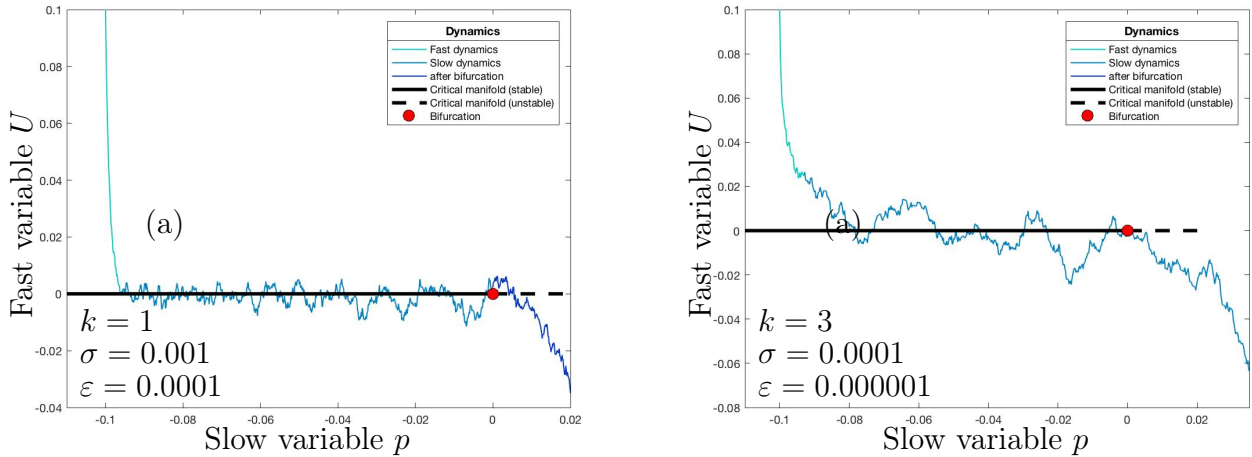
\begin{figure}
\begin{minipage}{0.5\linewidth}
\centering
\begin{overpic}[width=\textwidth]{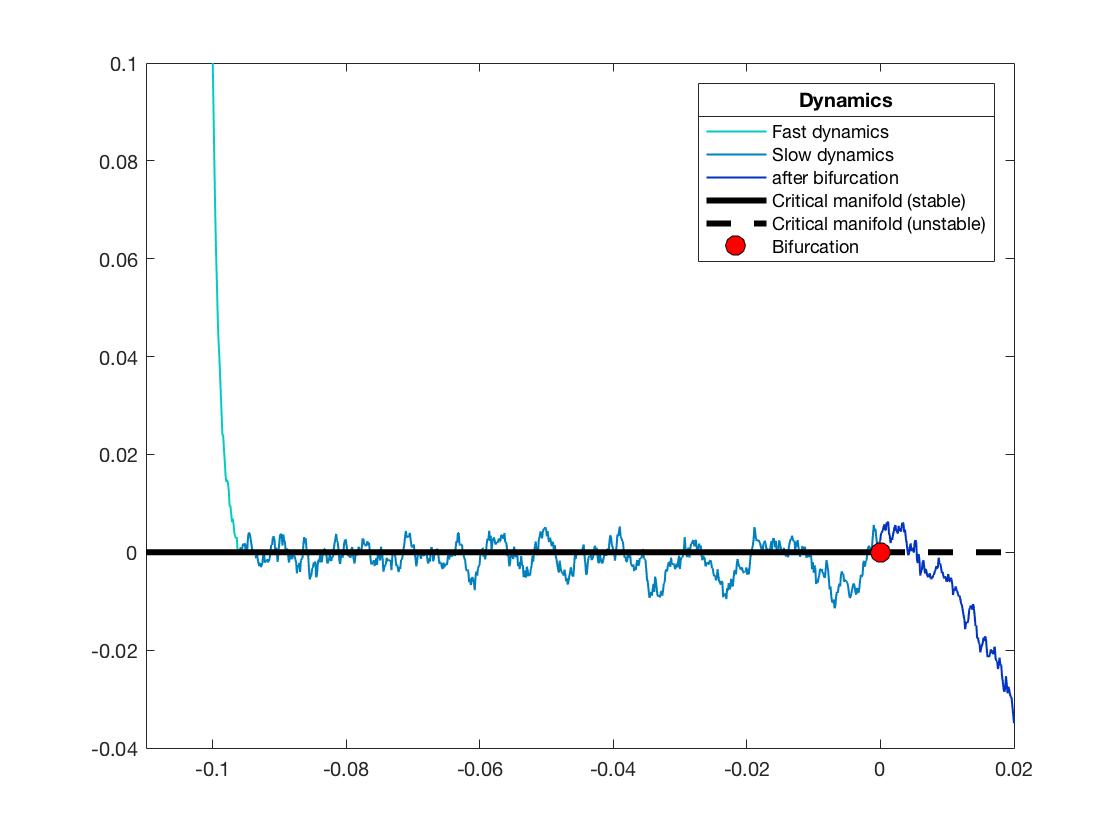}
\put(6,38){\makebox(0,0){\rotatebox{90}{Fast variable $U$}}}
\put(50,3){\makebox(0,0){\rotatebox{0}{Slow variable $p$}}}
\put(15,20){$k=1$}
\put(25,35){(a)}
\put(15,15){$\sigma=0.001$}
\put(15,10){$\varepsilon=0.0001$}
\end{overpic}
\end{minipage}
\hspace{0.2 cm}
\begin{minipage}{0.5\linewidth}
\centering
\begin{overpic}[width=\textwidth]{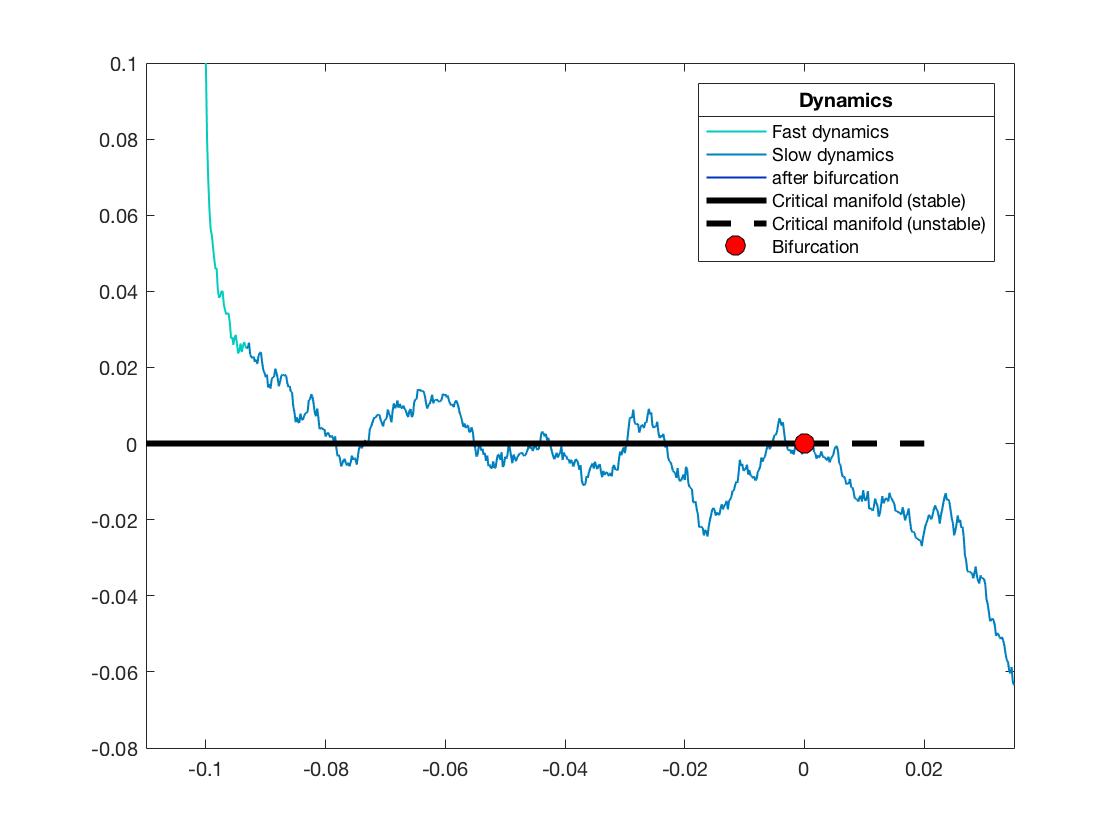}
\put(6,38){\makebox(0,0){\rotatebox{90}{Fast variable $U$}}}
\put(50,3){\makebox(0,0){\rotatebox{0}{Slow variable $p$}}}
\put(15,20){$k=3$}
\put(25,35){(a)}
\put(15,15){$\sigma=0.0001$}
\put(15,10){$\varepsilon=0.000001$}
\end{overpic}
\end{minipage}
\caption{\label{fig:ff3} (a) $k=1$, $(p_0,U_0)=(-0.1, 0.1)$, $\sigma=0.001$, 
$\varepsilon=0.0001$. (b) $k=3$, $(p_0,U_0)=(-0.1, 0.1)$, $\sigma=0.0001$, 
$\varepsilon=0.000001$.}
\end{figure}

Having defined the test model, we specify the setting, in which our predictions 
happen and the object we want to predict. Consider $U(s)$ and assume we have a 
time series of $w$ observations acquired at evenly spaced time intervals of 
length $\Delta s$, starting from time $s_{n-w+1}$ to time $s_n$. At time $s_n$ 
we want to predict, whether a bifurcation happens at a future time in the 
interval $[s_{n+\kappa- \delta}, s_{n+\kappa+ \delta}]$. We call $\kappa$ 
the lead time of our prediction, $\delta$ the uncertainty and $w$ the sliding 
window width. Given a single time series (or ``realization'') over the time 
window $[s_{n-w+1}, s_n]$, we approximate the variance via a sliding window 
estimate
$$v_n=\frac{1}{w} \sum_{i=n-w+1}^n (U(s_i)-\bar{U}(s_n))^2, 
\textrm{ where } \bar{U}(s_n) = \frac{1}{w} \sum_{i=n-w+1}^n U(s_i).$$
Qualitatively, the reason for using the sliding variance is that, if 
$\varepsilon$ is small enough, we can assume the parameter $p$ to be approximately 
constant in the sliding window. 

\begin{figure}
\centering
\begin{overpic}[width=0.8\textwidth]{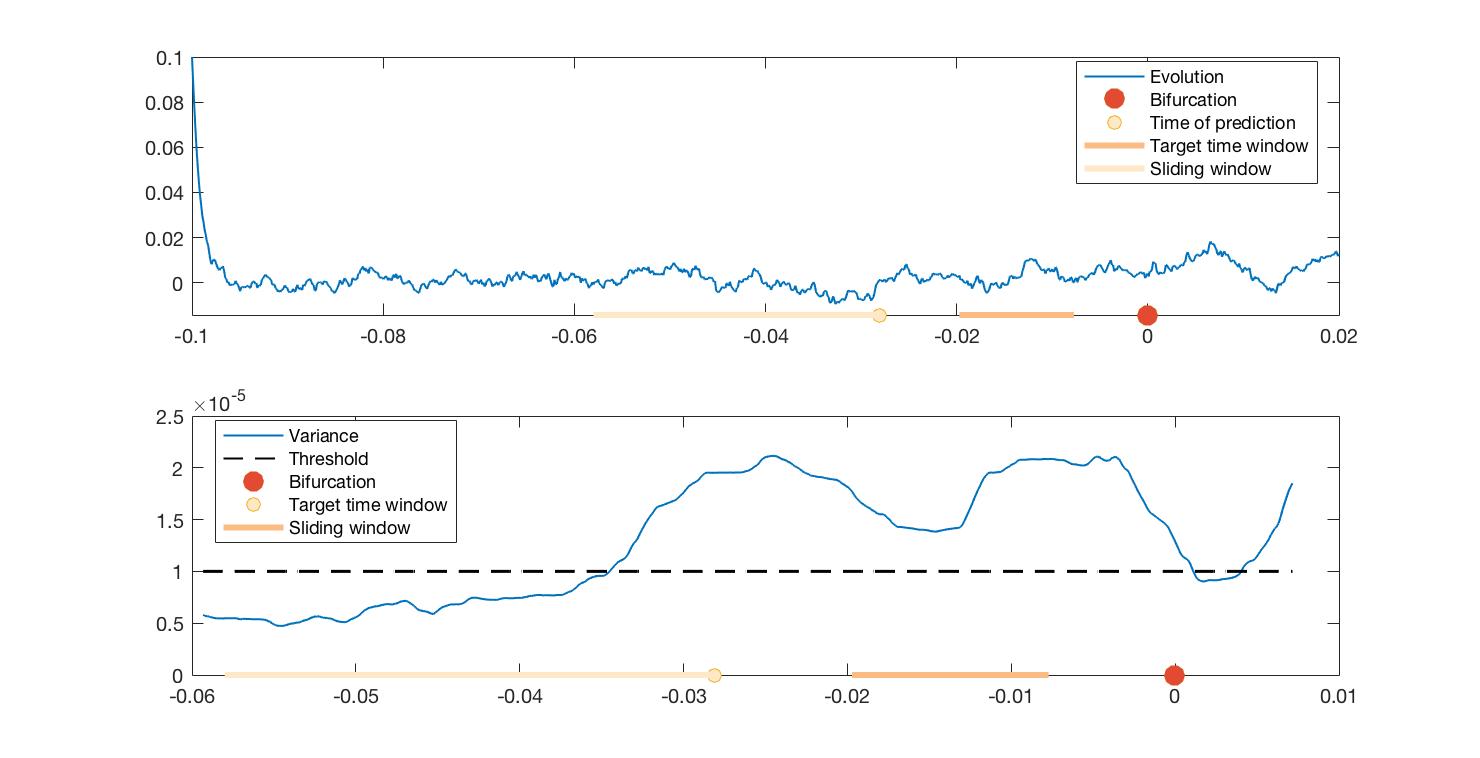}
\put(7,38){\makebox(0,0){\rotatebox{90}{Fast variable $U$}}}
\put(50,26){\makebox(0,0){\rotatebox{0}{Slow variable $p$}}}
\put(7,12){\makebox(0,0){\rotatebox{90}{Sliding variance}}}
\put(50,1){\makebox(0,0){\rotatebox{0}{Slow variable $p$}}}
\end{overpic}
\caption{\label{fig:predFP} An example of false positive prediction: the 
bifurcation happens outside the target interval and therefore the correct 
value of our estimator is $0$ (i.e.~no bifurcation); the plot below shows 
the values of the sliding variance and the sliding variance is above the 
threshold so our estimator wrongly classifies this as a positive prediction.}
\end{figure}
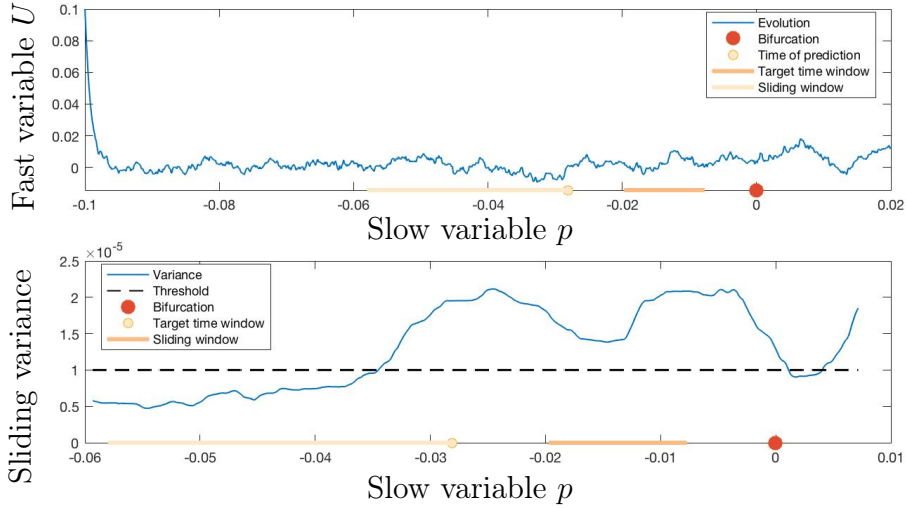

Now we define a family of binary estimators 
as follows: we raise an alarm for a tipping when the value of the variance 
goes above the threshold $d$. We define the indicator function for the alarm volume
as follows 
\begin{equation}
A_n(v_n, d) =
\begin{cases}
1~ \textrm{ if } v_n \geq d, \\
0~ \textrm{ otherwise.}
\end{cases}
\end{equation}
To clarify the prediction procedure we show two examples in 
Figure~\ref{fig:predFP} and Figure~\ref{fig:predTP}. In each example both 
the evolution of the fast-slow system and the sliding variance are shown. 
In particular, in the plot of the sliding variance we have highlighted the 
threshold level (black dashed line), the sliding window used to compute the 
variance (yellow), the interval $[s_{n+\kappa- \delta}, s_{n+\kappa+ \delta}]$ 
(orange) and the bifurcation point (red). The time at which the prediction is 
performed is marked with a yellow dot.  

\begin{figure}
\centering
\begin{overpic}[width=0.8\textwidth]{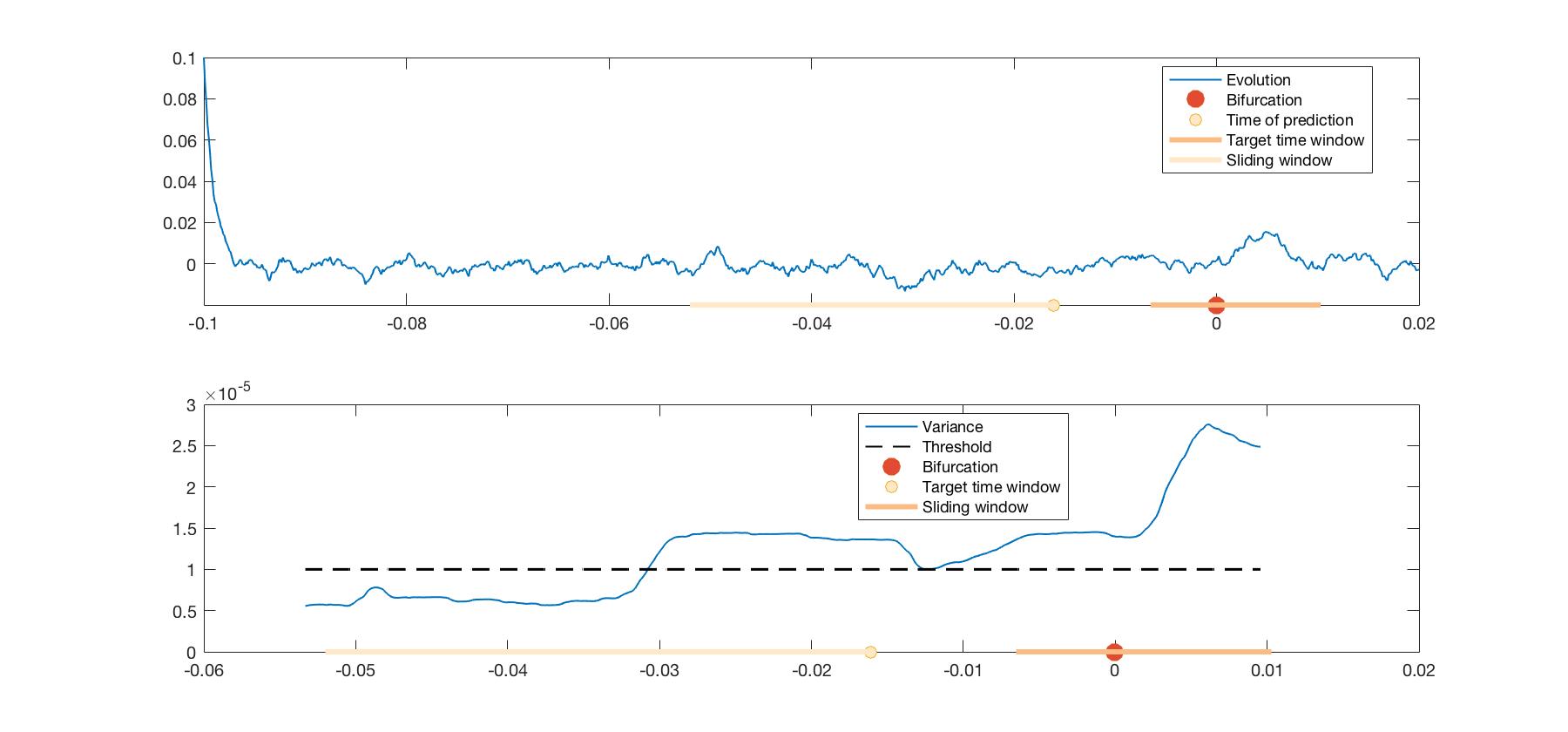}
\put(7,35){\makebox(0,0){\rotatebox{90}{Fast variable $U$}}}
\put(50,23){\makebox(0,0){\rotatebox{0}{Slow variable $p$}}}
\put(7,12){\makebox(0,0){\rotatebox{90}{Sliding variance}}}
\put(50,1){\makebox(0,0){\rotatebox{0}{Slow variable $p$}}}
\end{overpic}
\caption{\label{fig:predTP} An example of true positive prediction: the 
bifurcation happens inside the target interval and therefore the correct 
value of our estimator is $1$ (i.e.~bifurcation); the plot below shows that 
the sliding variance is above the threshold. Hence, our estimator correctly 
classifies this as a positive prediction.}
\end{figure}
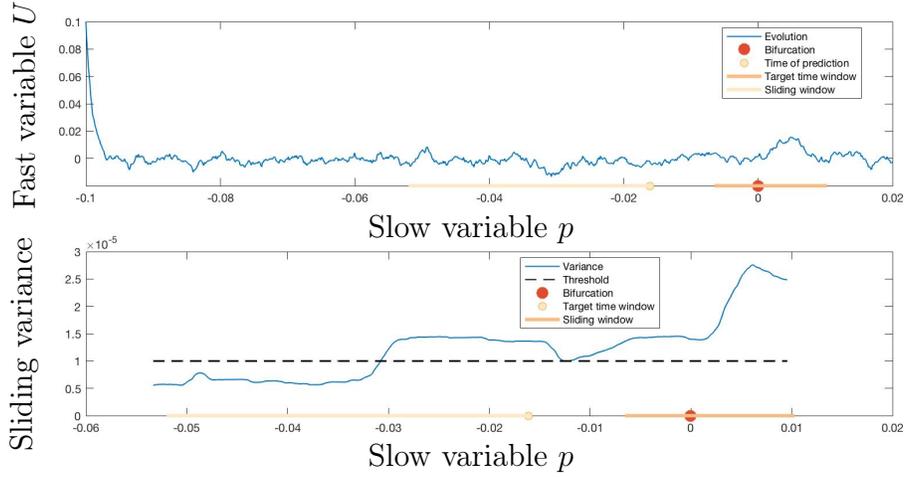

Now one defines \textit{true} and \textit{false positive rates} as
$$TPR(M)=\frac{\# \{ \textrm{x: x is true positive} \} }{\# 
\{ \textrm{x: x is positive} \}} \textrm{ and } FPR(M)=\frac{\# 
\{ \textrm{x: x is false positive} \} }{\# \{ \textrm{x: x is 
negative} \}}.$$ 
A standard way to represent the efficiency of the classifier is then 
to plot the graph FPR vs TPR. The space having FPR on the $x$-axis 
and TPR on the $y$-axis is called \textit{ROC space}. In ROC space 
a good classifier is very close to the point $(0,1)$, which represents 
the perfect classifier. Note also that the diagonal (i.e.~the line $FPR=TPR$) 
in the ROC space represents random guesses. Therefore, an obvious minimal 
requirement for the efficiency of a classifier is being represented 
above this line. Since our estimator depends on the threshold $d$, 
it will be represented as a curve in the ROC space, known as ROC curve.

\begin{figure}
\centering
\begin{overpic}[width=0.6\textwidth]{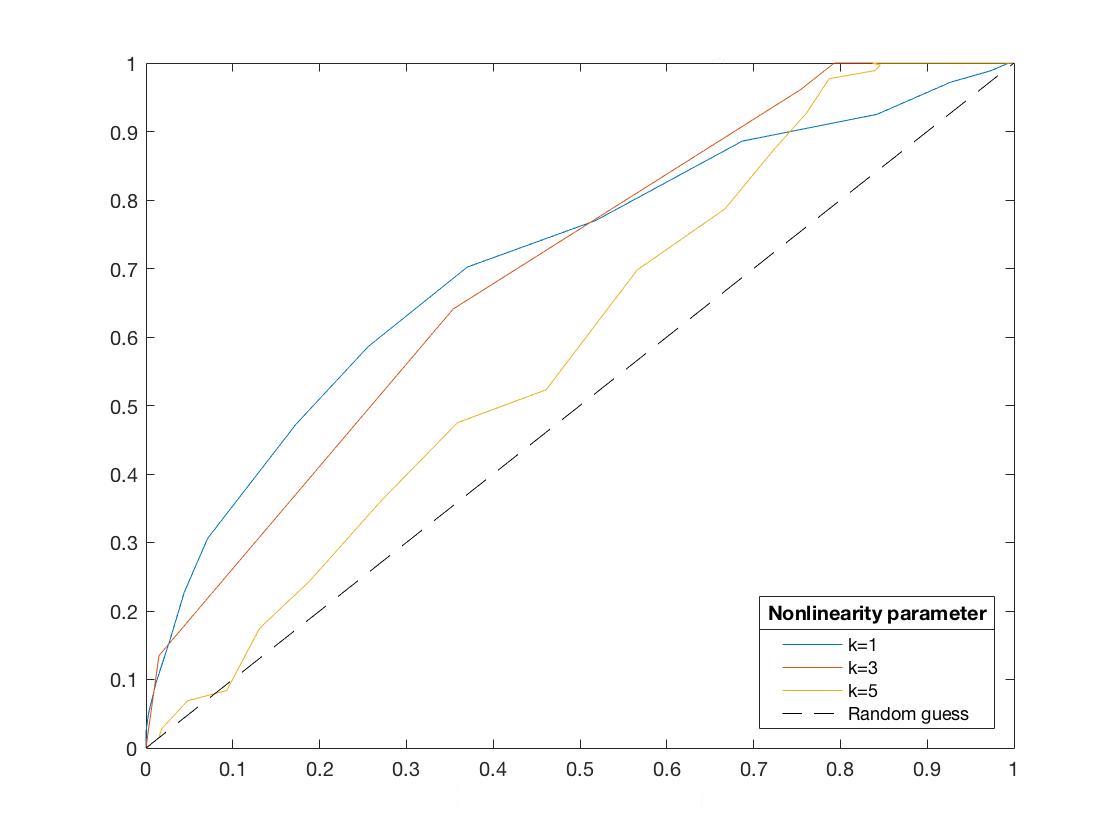}
\put(7,40){\makebox(0,0){\rotatebox{90}{True positive rate (TPR)}}}
\put(50,3){\makebox(0,0){\rotatebox{0}{False positive rate (FPR)}}}
\end{overpic}
\caption{\label{fig:finalgeneral} ROC curves for different 
values of the parameter $k$. Because the ROC curves are above 
the diagonal (the digonal corresponds to random guesses), the estimator 
we propose is statistically relevant.}
\end{figure}

We observe that the estimator we defined is statistically relevant for 
different values of the nonlinearity. In Figure~\ref{fig:finalgeneral} 
ROC curves are plotted for $k=1$, $k=3$ and $k=5$ and parameter values 
\begin{itemize}
 \item $[k, \kappa, \delta, w, \sigma, \varepsilon, N]=
[1,10,0,10^2,10^{-3},10^{-4}, 10^3]$, 
 \item $[k, \kappa, \delta, w, \sigma, \varepsilon, N]=
[3,10^3,500,3\cdot 10^3,10^{-5},10^{-8},10^5]$, 
 \item $[k, \kappa, \delta, w, \sigma, \varepsilon, N]=
[5,10^3,0,3\cdot 10^3,5\cdot10^{-6},2.5\cdot 10^{-9},10^4]$.
\end{itemize}
In the simulations we set the number of observations of a tipping even, i.e., the number of
sample paths, to $1000$. It is also interesting to fix a value of $k$ 
and study the efficiency 
of our estimator as a function of the sliding window $w$, the lead 
time $\kappa$ and the uncertainty $\delta$. Intuitively, we expect 
the efficiency of the estimator to be positively correlated to $w$ and 
$\delta$, but negatively correlated to $\kappa$. This reflects the fact 
that a higher availability of data, as well as the possibility to allow 
bigger uncertainty, improves our predictive ability. On the other hand, 
if we try to predict the bifurcation far in advance (i.e.~large lead 
time), we should obtain poorer results. 

\begin{figure}
\centering
\begin{overpic}[width=0.9\textwidth]{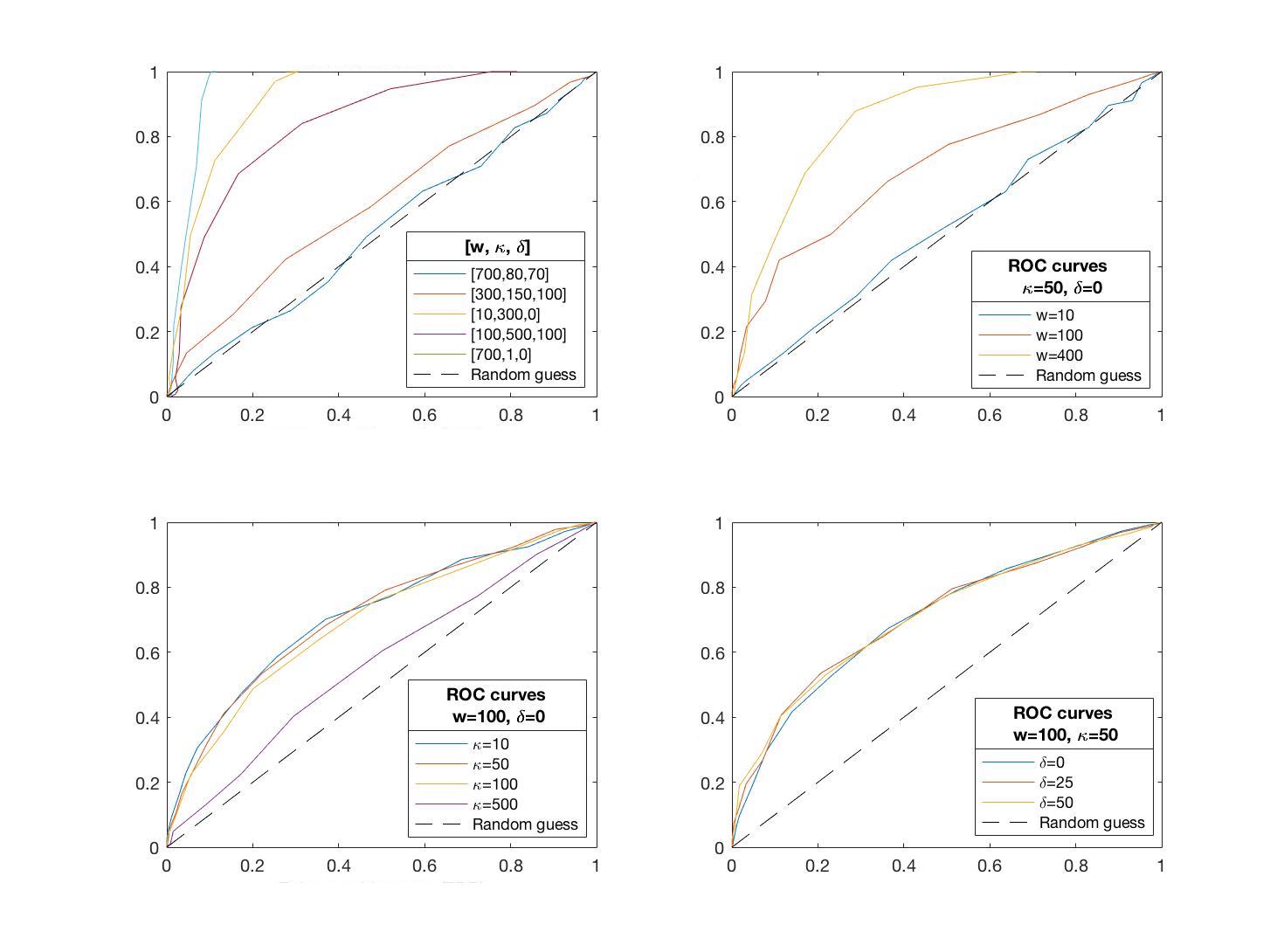}
\put(7,55){\makebox(0,0){\rotatebox{90}{True positive rate (TPR)}}}
\put(75,39){\makebox(0,0){\rotatebox{0}{False positive rate (FPR)}}}
\put(7,20){\makebox(0,0){\rotatebox{90}{True positive rate (TPR)}}}
\put(75,3){\makebox(0,0){\rotatebox{0}{False positive rate (FPR)}}}
\put(51,55){\makebox(0,0){\rotatebox{90}{True positive rate (TPR)}}}
\put(30,39){\makebox(0,0){\rotatebox{0}{False positive rate (FPR)}}}
\put(51,20){\makebox(0,0){\rotatebox{90}{True positive rate (TPR)}}}
\put(30,3){\makebox(0,0){\rotatebox{0}{False positive rate (FPR)}}}
\end{overpic}
\caption{\label{fig:roctre} The four plots show different ROC curves 
in the nonlinear case $k=3$. In the first plot ROC curves for different 
values of the parameters $w, \kappa, \delta$ are shown. The other plots 
show the comparison between different values of each parameter while the 
others remain fixed. These results have been obtained for $(U_0, p_0)
=(0.1;-0.1), \sigma=10^{-5}$ and $\varepsilon=10^{-8}$ averaging over 
$1000$ sample paths obtained using Euler-Maruyama method with a grid 
of $10000$ points.}
\end{figure}
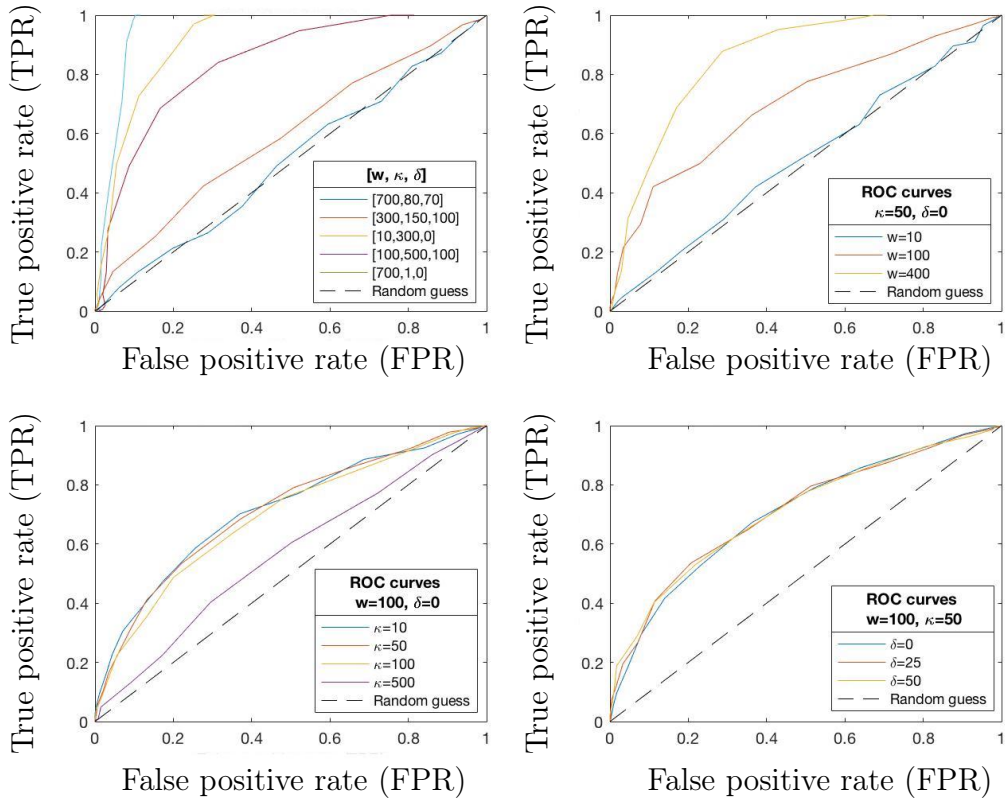

Figure~\ref{fig:roctre} shows the case $k=3$. Our expectations are 
confirmed by the data. A wider sliding window and a smaller lead time 
give better predictions. Surprisingly, the uncertainty $\delta$ seems 
to have no major impact on the results for the interval of 
values that we tested here. This hints at the conjecture that lead time
and sliding window width are the major limiting factors for the parameter
configurations we tested here.\medskip

\textbf{Acknowledgements:} CK would like to thank the VolkswagenStiftung for support
via a Lichtenberg Professorship Grant.

\end{document}